\documentclass[letterpaper,onecolumn,11pt,accepted=2021-07-27]{quantumarticle}
\pdfoutput=1
\usepackage[utf8]{inputenc}
\usepackage[english]{babel}
\usepackage[T1]{fontenc}
\usepackage[numbers,sort&compress]{natbib}
\usepackage{amsmath,amssymb,amsthm,amsfonts,bbm,mathtools,braket,microtype}
\usepackage{doi}
\usepackage{hyperref}
\usepackage{cleveref}

\newtheorem{theorem}{Theorem}
\newtheorem{lemma}[theorem]{Lemma}
\newtheorem{proposition}[theorem]{Proposition}

\theoremstyle{definition}
\newtheorem{definition}[theorem]{Definition}
\newtheorem{remark}[theorem]{Remark}
\newtheorem{notation}[theorem]{Notation}
\newtheorem{observation}[theorem]{Observation}
\newtheorem{fact}[theorem]{Fact}

\renewcommand{\Pr}{\mathop{\bf Pr\/}}
\newcommand{\E}{\mathop{\bf E\/}}

\newcommand{\R}{\mathbbm R}
\newcommand{\N}{\mathbbm N}
\newcommand{\F}{\mathbbm F}

\newcommand{\eps}{\epsilon}
\newcommand{\la}{\langle}
\newcommand{\ra}{\rangle}
\newcommand{\wt}[1]{\widetilde{#1}}
\newcommand{\ol}[1]{\overline{#1}}

\newcommand{\calS}{\mathcal{S}}

\newcommand{\by}{\boldsymbol{y}}
\newcommand{\bA}{\boldsymbol{A}}
\newcommand{\bC}{\boldsymbol{C}}
\newcommand{\bG}{\boldsymbol{G}}
\newcommand{\bH}{\boldsymbol{H}}
\newcommand{\bR}{\boldsymbol{R}}
\newcommand{\bT}{\boldsymbol{T}}

\newcommand{\zero}{\mathtt{0}}
\newcommand{\one}{\mathtt{1}}

\newcommand{\etaest}{\eta_{\mathrm{est}}}
\newcommand{\perr}{p_{\mathrm{err}}}
\newcommand{\estperr}{\widehat{p}_{\mathrm{err}}}

\crefformat{section}{\S#2#1#3} 
\crefformat{subsection}{\S#2#1#3}
\crefformat{subsubsection}{\S#2#1#3}
\crefrangeformat{section}{\S\S#3#1#4 to~#5#2#6}
\crefmultiformat{section}{\S\S#2#1#3}{ and~#2#1#3}{, #2#1#3}{ and~#2#1#3}

\begin{document}

\title{Pauli error estimation via Population Recovery}

\author{Steven T. Flammia}
\affiliation{AWS Center for Quantum Computing, USA}
\affiliation{IQIM, California Institute of Technology, USA}
\email{sflammi@amazon.com}

\author{Ryan O'Donnell}
\affiliation{Computer Science Department, Carnegie Mellon University, USA}
\email{odonnell@cs.cmu.edu}

\maketitle

\begin{abstract}
Motivated by estimation of quantum noise models, we study the problem of learning a Pauli channel, or more generally the Pauli error rates of an arbitrary channel. By employing a novel reduction to the ``Population Recovery'' problem, we give an extremely simple algorithm that learns the Pauli error rates of an $n$-qubit channel to precision~$\eps$ in $\ell_\infty$ using just $O(1/\eps^2) \log(n/\eps)$ applications of the channel. This is optimal up to the logarithmic factors. Our algorithm uses only unentangled state preparation and measurements, and the post-measurement classical runtime is just an $O(1/\eps)$ factor larger than the measurement data size. It is also impervious to a limited model of measurement noise where heralded measurement failures occur independently with probability $\le 1/4$.

We then consider the case where the noise channel is close to the identity, meaning that the no-error outcome occurs with probability $1-\eta$. In the regime of small $\eta$ we extend our algorithm to achieve \emph{multiplicative} precision~$1 \pm \eps$ (i.e., additive precision $\eps \eta$) using just $O(\frac{1}{\eps^2 \eta}) \log(n/\eps)$ applications of the channel.
\end{abstract}

\section{Introduction} \label{sec:intro}

A major challenge in the analysis of engineered quantum systems is estimating and modeling noise.
The most standard theoretical model for noise in the study of quantum error correction and fault tolerance~\cite{Ter15} is the $n$-qubit \emph{Pauli channel}:
\begin{equation}    \label{eqn:PC}
    \rho \mapsto \sum_{C \in \{0,1,2,3\}^n} p(C) \cdot \sigma_C \rho \sigma_C^\dagger.
\end{equation}
Here $\sigma_C = \sigma_{C_1} \otimes \cdots \otimes \sigma_{C_n}$ is a tensor product of the Pauli operators $\sigma_0, \sigma_1, \sigma_2, \sigma_3$, and $p$ is a probability distribution on $\{0,1,2,3\}^n$.
The numbers~$p(C)$ are referred to as the \emph{Pauli error rates}.
Additional motivation for the Pauli channel model comes from the practical technique of randomized compiling~\cite{Knill2005,WE16}, which converts a general noise channel~$\Lambda$ (with potentially coherent errors) to a Pauli channel~$\Lambda_P$ having the same process fidelity as the original channel.
We refer to the $p(C)$ values for $\Lambda_P$ as the ``Pauli error rates'' of the original general channel~$\Lambda$.

Given an experimental setup (possibly with randomized compiling), a natural challenge is to diagnose errors in the system via \emph{Pauli error estimation}.
Here the goal is to estimate the large Pauli error rates of an unknown channel by preparing states, passing them through the channel, and measuring them.
The main desideratum is to minimize the number of measurements; additionally one would like to use simple state preparation and measurement processes and minimal computational overhead.
We remark that full tomography for arbitrary $n$-qubit channels requires at least $4^n/\eps^2$ measurements, with more practical methods requiring at least~$8^n/\eps^2$.

In this work, we give very simple and efficient algorithms for learning all of the large Pauli error rates of an $n$-qubit channel.
Our first main result is the following:
\begin{theorem}
\label{thm:intro-pop}
    There is a learning algorithm that, given parameters $0 < \delta, \eps < 1$, as well as access to an $n$-qubit channel with Pauli error rates~$p$, has the following properties:
    \begin{itemize}
        \item It prepares $m = O(1/\eps^2)\cdot \log(\frac{n}{\epsilon \delta})$ unentangled $n$-qubit pure states, where each of the $mn$ $1$-qubits states is chosen uniformly at random from $\{\ket{0}, \ket{1}, \ket{+}, \ket{-}, \ket{i}, \ket{-i}\}$;
        \item It passes these $m$ states through the Pauli channel.
        \item It performs unentangled measurements on the resulting states, with each qubit being measured in either the $\{\ket{0}, \ket{1}\}$-basis, the $\{\ket{+}, \ket{-}\}$-basis, or $\{\ket{i}, \ket{-i}\}$-basis.
        \item It performs an $O(mn/\eps)$-time classical post-processing algorithm on the resulting $mn$ measurement outcome bits.
        \item It outputs hypothesis Pauli error rates $\widehat{p}$ in the form of a list of at most $\frac{4}{\eps}$ pairs $(C, \widehat{p}(C))$, with all unlisted $\widehat{p}$ values treated as~$0$.
    \end{itemize}
    The algorithm's hypothesis $\widehat{p}$ will satisfy $\|\widehat{p} - p\|_\infty \leq \eps$ except with probability at most~$\delta$.
\end{theorem}
Note that our ``sample complexity'' of $\wt{O}(1/\eps^2)$ is optimal up to the logarithmic term: The task of estimating Pauli error rates strictly (and vastly) generalizes the problem of estimating the bias of an unknown coin to additive precision~$\eps$ (and confidence~$1-\delta$), and this is known to require $\Theta(1/\eps^2) \cdot \log(1/\delta)$ coin flips.
For comparison of our bounds with previous work~\cite{FW20,Harper2020,Harper2021}, see \Cref{sec:prev}.

When the channel is modeling quantum noise, one hopes and expects that the nontrivial error rate, $\eta = 1 - p(0^n)$, is small.
In this case, a natural and more ambitious goal is to first estimate~$\eta$, and then to estimate all other Pauli error rates to \emph{multiplicative} precision~$1\pm \eps$; i.e., additive precision $\pm \eps \eta$.
(This ambition was also pursued in~\cite{FW20,Harper2020}.)
Here the ideal sample complexity would be $O(\frac{1}{\eps^2 \eta})$.\footnote{Again, one can compare the task to the vastly simpler one of estimating the face probabilities of a $6$-sided die that comes up ``$1$'' with probability $1-\eta$.
When rolling many times, one obtains a non-$1$ outcome roughly every~$1/\eta$ rolls.
Thus the task becomes very similar to estimating the face probabilities of a $5$-sided die to additive precision~$\eps$, but with a $1/\eta$ ``slowdown''.}
If one uses our \Cref{thm:intro-pop} as a black box, it would use~$\wt{O}(\frac{1}{\eps^2 \eta^2})$ measurements.
The extra factor of $1/\eta$ here is quite undesirable (as one might imagine a typical parameter setting to be something like $\eta = 10^{-2}$, $\eps = 10^{-1}$).
We show that it can be eliminated:
\begin{theorem}\label{thm:intro-pop2}
    In the setting of \Cref{thm:intro-pop}, suppose the overall error rate is $\eta = 1 - p(0^n)$.
    One can augment the algorithm so that, given in addition a ``noise floor'' parameter~$0 < \eta_0 < 1$, it has the following properties:
    \begin{itemize}
        \item It first makes at most $m_0 \coloneqq O(1/\eta_0) \cdot \log(1/\delta)$ measurements (as in \Cref{thm:intro-pop}).
        \item It does $O(m_0 n)$-time classical processing, then either outputs ``$\eta \leq \eta_0$'' and halts, or proceeds.
        \item It then operates as in \Cref{thm:intro-pop}, but makes $m \coloneqq O(\frac{1}{\eps^2 \eta}) \cdot \log(\tfrac{n}{\eps \delta})$ measurements.
      \end{itemize}
    Its outputs are correct, with a guarantee of $\|\widehat{p} - p\|_\infty \leq \eps \eta$, except with probability at most~$\delta$.
\end{theorem}

Finally, we show that our algorithm can be made impervious to a limited amount of measurement noise.
Specifically, suppose that our measuring devices have the following property:
When measuring a $1$-qubit state from $\{\ket{0}, \ket{1}, \ket{+}, \ket{-}, \ket{i}, \ket{-i}\}$ in one of the bases $\{\ket{0}, \ket{1}\}$, $\{\ket{+}, \ket{-}\}$, or $\{\ket{i}, \ket{-i}\}$, the device fails (reading out~``$\texttt{?}$'') with probability~$\nu$, and otherwise behaves ideally.
We assume that the failures are independent, and that the algorithm may know the parameter~$\nu$ (thanks to prior estimation).
In this case, we will see that it is almost automatic to obtain the following extension:
\begin{theorem}                                     \label{thm:pop4}
    \Cref{thm:intro-pop2} continues to hold for any any constant $\nu \leq \frac14$.
\end{theorem}
For the more challenging task of handling general SPAM (state preparation and measurement) error, see the discussion in \Cref{sec:prev}.

\subsection{Techniques}
Our algorithm employs a novel reduction from Pauli error estimation to the task in classical unsupervised learning known as \emph{Population Recovery}.
Population Recovery was introduced by Dvir, Rao, Wigderson, and Yehudayoff in 2012~\cite{DRWY12}, and has been studied in numerous subsequent works~\cite{BIMP13,MS13,LZ15,WY16,DST16,LZ17,PSW17,BCSS19,BCFSS19,DOS20,Nar20}.
A Population Recovery problem is specified by a \emph{classical channel}~$\calS$ --- i.e., a stochastic map $\calS : \Sigma \to \Gamma$ for some finite alphabets $\Sigma, \Gamma$.
The task is to learn an unknown probability distribution~$p$ on $\Sigma^n$ to $\ell_\infty$-error~$\eps$, with the twist being that samples are mediated by the channel.
That is, when the learner requests a sample, first $x \in \Sigma^n$ is drawn according to~$p$, but then only $y = \Sigma(x_1)\Sigma(x_2)\cdots\Sigma(x_n)$ is revealed to the learner.
The most well-studied cases are the binary symmetric channel and the binary erasure channel, the former being noticeably more challenging; lately, the deletion channel has also begun to be studied.
(Each of these channels also requires specifying the crossover/erasure/deletion probability~$r$.)

Our work shows how to efficiently convert the Pauli error estimation task to that of Population Recovery with respect to the so-called \emph{binary $Z$-channel} with crossover probability~$\frac13$.
This is the channel with $\Sigma = \Gamma = \{\zero,\one\}$ in which $\zero$'s are ``transmitted'' correctly, but $\one$'s are flipped to~$\zero$ with probability~$\frac13$.
We observe that the known methods for Population Recovery with respect to the binary erasure channel with erasure probability~$r$ also apply equally well to the $Z$-channel with crossover probability~$r$.
We then use the fact that there is a known, highly efficient Population Recovery algorithm for erasures with probability at most~$\frac12$.~\cite{DRWY12,MS13,DOS20,PSW17}
(Indeed, the fact that even probability~$\frac12$ can be tolerated is the reason our Pauli error estimation algorithm can handle additional measurement noise as in \Cref{thm:pop4}.)

\subsection{Previous and related work} \label{sec:prev}

There are several prior works that study the estimation of so-called \textit{generalized} Pauli channels, which act on a $d$-dimensional quantum system and do not contain explicit tensor product structure. 
These works typically (though not always~\cite{ur_Rehman_2021}) make the much stronger assumption that ideal entangled states can be prepared to assist the channel estimation, and they focus on finding efficient estimators that saturate the Cram\'{e}r-Rao bound. 
Fujiwara and Imai first showed that entanglement-assisted channel estimation of generalized Pauli channels could achieve the Cram\'{e}r-Rao bound~\cite{FI2003}, though much simpler proofs of such theorems are now available~\cite[Exer.~6.51--6.54]{Hayashi2017}. 
In particular, these results show that entanglement-assisted estimation of generalized Pauli channels can be done with a sample complexity of $O(1/\eps^2)$ in the $\ell_\infty$ norm. 

What about the case of $n$-qubit Pauli channels in this entanglement-assisted setting? 
An $\eps$-close estimate in the $\ell_\infty$ norm is also achievable with only $O(1/\eps^2)$ samples in this setting. 
This can be seen by noting that inputting half of a maximally entangled state into a Pauli channel and measuring in a Bell basis gives completely distinguishable outcomes for each Pauli error%
\footnote{This is essentially superdense coding~\cite{Bennett1992}. 
One can show by computing the diamond norm of the difference between two Pauli channels that this strategy has optimal sample complexity up to a constant factor.}.
The problem therefore reduces to estimating a \textit{classical} probability distribution on $4^n$ outcomes, and for this well-studied problem the sample complexity is well known to be $\Theta(1/\eps^2)$ (see Ref.~\cite{Canonne2020} for a simple proof). 
In light of this, one way to interpret our~\Cref{thm:intro-pop} is that entanglement-free estimation of Pauli channels is at most a logarithmic factor away from the optimal sample complexity, at least for the $\ell_\infty$ norm. 

The problem of Pauli error estimation for $n$-qubit channels without entanglement was first studied in depth in work of the first author and Wallman~\cite{FW20}.
It is not possible to directly compare those results with ours, for several reasons.
The most immediate reason is that their complexity bounds typically include a factor of~$\wt{O}(1/\Delta)$, where ``$\Delta$'' is another parameter, the spectral gap of the Pauli channel being learned.
We have $\Delta \le 2\eta$, where $\eta = 1 - p(0^n)$ is the nontrivial error rate, and this is saturated in the most favorable case.
However, in general $\Delta$ may be arbitrarily small, or even zero, for relatively simple channels.
In practice, a user of the algorithm in \cite{FW20} would set a spectral cutoff $\Delta_0$ and allow estimation errors for channel eigenvalues in the interval $(1-\Delta_0,1]$, but no analysis is done in~\cite{FW20} of the extra error incurred by this cutoff.
Thus, in the worst case, their results as formally stated do not give any guarantee.

On the other hand, the results of~\cite{FW20} are impervious to a much more challenging model of measurement error (``SPAM'').
This model imposes that before the learner measures the channel's output, an \emph{additional} unknown channel~$\Xi$ is applied to the state.
(It is assumed that $\Xi$ satisfies the extremely mild condition that its nontrivial error rate is bounded away from~$1$.)
It might seem impossible to disentangle~$\Xi$ from the main channel~$\Lambda$ to be learned, but the authors of~\cite{FW20} use the fact that one is at liberty to pass a state~$\rho$ through $\Lambda$ several times (say, $k$ times) before it is subjected to~$\Xi$; i.e., the learner may obtain $\Xi \Lambda^k \rho$ for $\rho$ and $k \in \N$ of the learner's choosing.
By carefully choosing~$k$ values up to $O(1/\Delta)$, the authors of~\cite{FW20} show that~$\Xi$ can essentially be expunged.
(Note that, in practice, multiple uses of the channel are often far less costly than even a single measurement.)

Finally, the first algorithm in \cite{FW20} judges its hypothesis with respect to the $\ell_2$-norm, rather than the $\ell_\infty$ norm as in this paper.
This distinction is relatively minor, however, as the norms are roughly equivalent for probability distributions: $\|\widehat{p} - p\|_\infty \leq  \|\widehat{p} - p\|_2  \leq \|\widehat{p} - p\|^{1/2}_\infty$, and one may refine this further to take into account dependence on $\eta = 1 - p(0^n)$.

With these caveats, we state (simplifications of) the relevant main results in~\cite{FW20}:
\begin{theorem}[\cite{FW20}]
    There exists a SPAM-tolerant algorithm that makes $\wt{O}(2^n \log(1/\Delta))/\eps^2$ measurements, with $O(1/\Delta)$ channel-uses per measurement, and with high probability outputs an estimate~$\widehat{p}$ of the channel's Pauli error rates~$p$ satisfying $\|\widehat{p} - p\|_2 \leq \eps \eta$.
\end{theorem}
In the favorable case of $\Delta = \Theta(\eta)$, this is somewhat comparable to our \Cref{thm:intro-pop2}; the above theorem has much better SPAM-tolerance, but a complexity that is greater by roughly~$2^n$.

The authors of~\cite{FW20} also present a heuristic for identifying a set~$S$ corresponding to large Pauli error rates with the following guarantee.
\begin{theorem}[\cite{FW20}]
    For any set $S \subseteq \{0,1,2,3\}^n$, there exists a SPAM-tolerant algorithm that makes $\wt{O}(\log |S|)\log \log(1/\Delta) /\eps^4$ measurements, with $O(1/\Delta)$ channel-uses per measurement, and with high probability outputs estimates~$\widehat{p}(C)$ for each $C \in S$ satisfying $|\widehat{p}(C) - p(C)| \leq \eps \eta$.
\end{theorem}
However, no guarantee is proven that the set $S$ will contain the $|S|$ largest error rates. 

The results in~\cite{Harper2021} are also somewhat incomparable to the present paper.
The authors analyze Pauli channels with a recovery guarantee in the $\infty$-norm, but under the assumption that the Pauli channel has sparse and random support, and that the nonzero error rates are not too small (greater than some fixed $\eps_0$).
While the sparsity assumption is not critical in that analysis (the algorithm will approximate error rates smaller than $\eps_0$ as zero with high probability), the random support assumption is used in an essential way.
This is an undesirable assumption since it is very unlikely to hold in practice.\footnote{Perhaps surprisingly, the algorithm performs well on real data despite grossly violating this assumption~\cite{Harper2021}.}
The sample complexity is also not stated directly in terms of quantum measurements, but rather in terms of queries to a ``noisy eigenvalue oracle'' with Gaussian noise.
While this noisy oracle can be approximated by quantum measurements and finite sample complexity, quantum noise is not exactly Gaussian, so no direct comparison with the present work is possible without further analysis.

We remark that the techniques used in~\cite{FW20,Harper2021} are Fourier-based, and the heuristic from~\cite{FW20} described above is similar to the Goldreich--Levin learning algorithm~\cite{GL89}. 
In \Cref{sec:fourier}, we give an alternate Fourier-based approach to Pauli error estimation, one that is equivalent to our Population Recovery method ``in disguise''; in fact, the Goldreich--Levin algorithm becomes equivalent to the Individual-to-Population Recovery reduction!

It is our belief that these Fourier techniques can actually be used to provide a common generalization of the results of this paper and of~\cite{FW20}; i.e., efficient SPAM-tolerant Pauli error estimation with no dependence on~$\Delta$.
We leave this for future work.

\section{Notation}

\begin{notation}
    The \emph{$1$-qubit Pauli matrices} are the unitary, hermitian matrices
    \[
        \sigma_0 = \begin{pmatrix} 1 & 0 \\
                                                           0 & 1 \end{pmatrix}, \qquad
        \sigma_1 = \sigma_x = \begin{pmatrix} 0 & 1 \\
                                                           1 & 0 \end{pmatrix}, \qquad
        \sigma_2 = \sigma_y = \begin{pmatrix} 0 & -i \\
                                                           i & \phantom{+}0 \end{pmatrix}, \qquad
        \sigma_3 = \sigma_z = \begin{pmatrix} 1 & \phantom{+}0 \\
                                                         0 & -1 \end{pmatrix}.
    \]
    As operators on the Bloch sphere, $\sigma_1, \sigma_2, \sigma_3$ act as rotations by~$\pi$ about the $1$st, $2$nd, $3$rd axis (aka $x$-, $y$-, $z$-axis), respectively.
    More generally, an \emph{$n$-qubit Pauli matrix}, indexed by string $A \in \{0,1,2,3\}^n$, is
    $
        \sigma_A = \bigotimes_{j=1}^n \sigma_{A_j}.
    $
\end{notation}

\begin{notation}
    For $a, b \in \{0,1,2,3\}$, there is some $c \in \{0,1,2,3\}$ such that $\sigma_a \sigma_b = \sigma_c$, up to a global phase.
    We introduce the notation $a \oplus b$ (equivalently, $b \oplus a$) for this~$c$; so, e.g, $1 \oplus 3 = 2$, $0 \oplus b = b$, etc.
    We extend the notation coordinate-wise: if $A, B \in \{0,1,2,3\}^n$, then $A \oplus B = (A_1 \oplus B_1, \dots, A_n \oplus B_n) \in \{0,1,2,3\}^n$ (and so $\sigma_A \sigma_B = \sigma_{A \oplus B}$, up to a global phase).
\end{notation}

\begin{notation}
    We write the orthonormal eigenbasis for the Pauli operator $\sigma_x$ as $\ket{\chi^1_{+}}, \ket{\chi^1_{-}}$.
    On the Bloch sphere these are the two unit vectors pointing in the positive (respectively, negative) direction along the $1$st ($x$-)axis; they are often called \mbox{$\ket{+}$, $\ket{-}$}.
    We use similar notation $\ket{\chi^2_{+}}, \ket{\chi^2_{-}}$ (often called $\ket{i}, \ket{-i}$) and $\ket{\chi^3_{+}}, \ket{\chi^3_{-}}$ (often called $\ket{0}, \ket{1}$) for $\sigma_2$ and $\sigma_3$.
\end{notation}

\begin{notation}    \label{not:commute-anticommute}
    For $a, b \in \{0,1,2,3\}$ we have that $\sigma_b \ket{\chi^a_+}$ is (up to a phase) $\ket{\chi^a_{\pm}}$, with the subscript being~${+}$ if $\sigma_a$ and $\sigma_b$ commute, and~${-}$ if $\sigma_a$ and $\sigma_b$ anticommute.
    To capture this, it will be convenient to introduce the following notation:
    \[
        a \star b = b \star a = \begin{cases}
                                                \zero & \text{if } |\{a, b, a \oplus b\}| < 3, \text{ i.e., } \sigma_a, \sigma_b \text{ commute;} \\
                                                \one & \text{if } |\{a, b, a \oplus b\}| = 3, \text{ i.e., } \sigma_a, \sigma_b \text{ anticommute.}
                                        \end{cases}
    \]
    Thus $\sigma_b \ket{\chi^a_{+}} = \ket{\chi^a_{(-1)^{a \star b}}}$ (up to a phase).
    We extend this notation coordinate-wise, writing $A \star B = (A_1 \star B_1, \dots, A_n \star B_n) \in \{\zero, \one\}^n$ for $A, B \in \{0,1,2,3\}^n$.
    For example,
    $
        (0,0,3,2,1) \star (3,1,1,2,2) = (\zero, \zero, \one, \zero, \one).
    $
\end{notation}

\begin{fact}\label{fact:dot-product}
    If we identify $\{0, 1, 2, 3\}$ with $\F_2^2$ by writing numbers in base~$2$, then $\oplus$ corresponds to the usual vector addition in $\F_2^2$, and $\star$ corresponds to the ``symplectic'' product: $a \star b = (a_1, a_2) \star (b_1, b_2) = a_1b_2 + a_2b_1$.
    This lets us see that $a \star (b \oplus c) = (a \star b) + (a \star c) \bmod 2$.
\end{fact}

\begin{notation}
    For a quantity $x$, we denote an estimate of $x$ by $\widehat{x}$. 
    We use boldface font (e.g., $\bA$) to denote a random variable. 
    If $\bA$ is drawn from the distribution $p$ we denote this by $\bA\sim p$, and let $A$ denote a concrete assignment to the variable $\bA$. 
    Addition (of scalars or vectors) modulo 2 is denoted $+_2$. 
    The Fourier transform of $f$ is denoted $\tilde{f}$. 
\end{notation}

\section{Learning a Pauli channel}
In this section we describe the basic setup for learning a Pauli channel.  
Learning the Pauli error rates of a general channel will end up being just a minor extension, discussed in \Cref{sec:general-chan}.

As described in \Cref{eqn:PC}, an $n$-qubit Pauli channel is determined by a probability distribution~$p$ on $\{0,1,2,3\}^n$.
This probability distribution induces the mixed unitary channel in which $\sigma_C$ is applied with probability~$p(C)$.
An $n = 5$ example:
\[
    p(00321) = 2/10, \quad
    p(01300) = 3/10, \quad
    p(11323) = 2/6, \quad
    p(30000) = 1/6, \quad
    p(C) = 0 \text{ otherwise.}
\]
We anthropomorphize by imagining a character Charlie who operates the channel; on receiving a state~$\rho$, Charlie first (secretly) draws $\bC \sim p$, then outputs the state $\sigma_{\bC} \rho$.

Alice the Learner would like to estimate the probability distribution~$p$ via interactions with Charlie.
Alice has the ability to prepare $n$-qubit states, to ``query'' Charlie (i.e., pass an $n$-bit state through his channel), and to measure states that she receives back.
Her goal is to learn a precise approximation to~$p$ (with high probability), while minimizing the number of queries to Charlie.

\begin{definition}  \label{def:nontrivial-probe}
    We say that Alice performs a \emph{nontrivial probe} if she does the following:
    \begin{itemize}
        \item She chooses a string $A \in \{1,2,3\}^n$.
        \item She prepares the (unentangled) $n$-qubit state $\ket{\psi_A}$ in which the $j$th qubit is $\ket{\chi^{A_j}_{+}}$.
        \item She passes $\ket{\psi_A}$ through Charlie, obtaining $\sigma_{C} \ket{\psi_A}$ with probability~$p(C)$.
        \item She does a (non-entangled) measurement on the resulting $n$-qubit state, measuring the $j$th qubit in the basis $\ket{\chi^{A_j}_{\pm}}$.
        \end{itemize}
\end{definition}

Continuing our $n = 5$ example, if Alice does a nontrivial probe with the string $A = 31122$, this entails preparing and passing to Charlie the state
\[
    \ket{\psi_{31123}} = \ket{\chi^3_{+}}\ket{\chi^1_{+}}\ket{\chi^1_{+}}\ket{\chi^2_{+}}\ket{\chi^2_{+}} \quad \Bigl(= \ket{0} \ket{+} \ket{+} \ket{i} \ket{i}\Bigr),
\]
and then measuring the $5$ returned qubits in the bases $\ket{\chi^3_{\pm}}$, $\ket{\chi^1_{\pm}}$, $\ket{\chi^1_{\pm}}$, $\ket{\chi^2_{\pm}}$, $\ket{\chi^2_{\pm}}$, respectively.

Now suppose that Charlie drew $C = 00321$ (which occurs with probability $2/10$ in our example).  Then the state returned to Alice would be
\begin{align*}
    (\sigma_0 \otimes \sigma_0 \otimes \sigma_3 \otimes \sigma_2  \otimes \sigma_1) \ket{\psi_{31122}}
    &= (\sigma_0 \ket{\chi^3_+}) \otimes (\sigma_0 \ket{\chi^1_+}) \otimes (\sigma_3 \ket{\chi^1_+}) \otimes (\sigma_2 \ket{\chi^2_+}) \otimes(\sigma_1 \ket{\chi^2_+}) \\
    &= e^{i \theta} \cdot \ket{\chi^3_+}  \ket{\chi^1_+}  \ket{\chi^1_{-}}  \ket{\chi^2_+}  \ket{\chi^2_{-}}
\end{align*}
for some phase $e^{i \theta}$ ($\theta \in \R$) that we did not bother to compute.
Now when Alice measures in the bases $\ket{\chi^3_{\pm}}$, $\ket{\chi^1_{\pm}}$, $\ket{\chi^1_{\pm}}$, $\ket{\chi^2_{\pm}}$, $\ket{\chi^2_{\pm}}$, her readout will, with probability~$1$, be
\[
     \ket{\chi^3_+}  \ket{\chi^1_+}  \ket{\chi^1_{-}}  \ket{\chi^2_+}  \ket{\chi^2_{-}}.
\]
The subscripts ${+},{+},{-},{+},{-}$ here are the $5$ bits of information conveyed to Alice by the readout, and we may think of instead labeling them as $\zero \zero \one \zero \one$ in accordance with \Cref{not:commute-anticommute}.
With this relabeling convention, we obtain:
\begin{fact}    \label{fact:readout}
    Suppose Alice performs a nontrivial probe with string $A \in \{1,2,3\}^n$, and suppose the random string drawn by Charlie is $C \in \{0,1,2,3\}^n$.
    Then when Alice measures, she obtains the readout $R = A \star C \in \{\zero, \one\}^n$.
\end{fact}

\begin{remark}
    So far we have pictured Alice as first choosing~$A$, and then Charlie as drawing a random~$C$.
    It is useful now to make a slight shift in perspective: for each interaction between Alice and Charlie, we will equivalently think of \emph{Charlie} as first (secretly) drawing~$C$, and then Alice gaining some partial information about this~$C$ by ``probing'' it using an~$A$ of her choice.
    We emphasize that Alice must make her choice of~$A$ without knowing the channel outcome~$C$.
\end{remark}

We now describe a trick that Alice may employ in probing the channel:

\begin{definition}
    For a channel distribution $p$ on $\{0,1,2,3\}^n$, and any fixed $B \in \{0,1,2,3\}^n$, define the \emph{$B$-altered channel distribution} $p^{\oplus B}$ on $\{0,1,2,3\}^n$ via $p^{\oplus B}(C) = p(B \oplus C)$.
\end{definition}

For any string $B \in \{0,1,2,3\}^n$ of her choosing, Alice can effectively simulate access to the $B$-altered channel:
If she wishes to simulate passing $\ket{\phi}$ through the $B$-altered channel, she could instead simply pass $\sigma_B \ket{\phi}$ through Charlie's actual channel.
(This may introduce a ``wrong'' global phase, but it doesn't matter for any measurement behavior that we consider here.)
But in fact, something even simpler is true:

\begin{observation} \label{obs:alter-trick}
    Given $B \in \{0,1,2,3\}^n$, if Alice wants to perform a nontrivial probe of $p^{\oplus B}$ based on string~$A$, she can pass $\ket{\psi_A}$ to Charlie as always.
    Then, when she measures and obtains $A \star C$, she can ``reinterpret'' this readout by adding in, mod~$2$, the string $A \star B \in \{\zero, \one\}^n$ (which she knows).
    Recalling \Cref{fact:dot-product}, this gives her $(A \star B) +_2 (A \star C) = A \star (B \oplus C)$.
    Thus the reinterpreted readout is indeed distributed as what she would get by probing $p^{\oplus B}$ with~$A$.
\end{observation}

A natural strategy for Alice is to make \emph{random} nontrivial probes.
It is easy to see the following:
\begin{fact}                                            \label{fact:z-channel}
    Fix a draw $C \in \{0,1,2,3\}^n$ for Charlie.
    Now if Alice performs a nontrivial probe with a uniformly random $\bA \in \{1,2,3\}^n$, then the coordinates of her readout $\bR = \bA \star C \in \{\zero,\one\}^n$ will be independent, with the following distribution for each $1 \leq j \leq n$:
    \begin{itemize}
        \item If $C_j = 0$ then $\bR_j$ will be $\zero$ with probability~$1$.
        \item If $C_j \neq 0$ then $\bR_j$ will be $\zero$ with probability $\frac13$ and $\one$ with probability $\frac23$.
    \end{itemize}
\end{fact}

We can state this more succinctly by introducing some additional terminology:
\begin{notation}
    For $B,C \in \{0,1,2,3\}^n$, define the string $C^{\neq B} \in \{\zero, \one\}^n$ by
    \[
        (C^{\neq B})_j = \begin{cases}
                                       \one & \text{if $C_j \neq B_j$,} \\
                                       \zero & \text{if $C_j = B_j$}.
                                    \end{cases}
    \]
\end{notation}
\begin{definition}
    Recall from information theory the so-called \emph{$Z$-channel with crossover probability~$r$}: it is the binary channel that leaves $\zero$~untouched and flips $\one$ to $\zero$ with probability~$r$.
\end{definition}

Now \Cref{fact:z-channel} can be restated as follows:
\begin{fact}                                            \label{fact:z-channel2}
    Fix a draw $C \in \{0,1,2,3\}^n$ for Charlie.
    Now if Alice performs a random nontrivial probe, her readout is the result of passing $C^{\neq 0^n}$ through a $Z$-channel with crossover probability~$\frac13$.
\end{fact}

\begin{observation}             \label{obs:altered-z}
    By combining \Cref{obs:alter-trick} with \Cref{fact:z-channel2}, we obtain the following:
    Fix a draw $C \in \{0,1,2,3\}^n$ for Charlie and suppose Alice performs a random nontrivial probe.
    She can then --- for any fixed $B \in \{0,1,2,3\}^n$ --- interpret her readout as $C^{\neq B}$ passed through a $Z$-channel with crossover probability~$\frac13$.
    \emph{Warning:} these reinterpretations are completely \emph{dependent}; she of course cannot get the result of \emph{independent} channel applications for various~$B$'s, unless she makes multiple probes.
\end{observation}

\section{Population Recovery}
\label{sec:poprecov}

With \Cref{obs:altered-z} in hand, we have effectively reduced the problem of learning a Pauli channel to a ``Population Recovery''-type problem (with a quantum-free definition).
To recap:  there is an unknown probability distribution $p$ on $\{0,1,2,3\}^n$, a learner may request samples, and when a sample $C$ is drawn from~$p$, the learner receives a binary string which can be interpreted as ``$C^{\neq B}$~passed through a $Z$-channel with crossover~$\frac13$'' for any $B \in \{0,1,2,3\}^n$ of the learner's choosing.

In this section we will give a solution to this problem that has optimal sample complexity (except possibly up
to a logarithmic factor) using techniques from the field of Population Recovery.
Our solution will immediately imply \Cref{thm:intro-pop} in the special case where the channel to be learned is indeed a Pauli channel.
The case of learning a \emph{general} channel's Pauli error rates is treated in \Cref{sec:general-chan}.
We remark that our Pauli channel algorithm only uses nontrivial probes, and thus only involves preparing the states $\ket{0}$, $\ket{+}$, and~$\ket{i}$.
The other three states $\ket{1}$, $\ket{-}$, and~$\ket{-i}$ are only used for the extension to general channels.

\paragraph{Idea of our solution.}
Using known techniques from Population Recovery, one can first reduce to the simpler task of ``Individual Recovery'' (estimating a single $p(B)$ value) via a coordinate-by-coordinate learning algorithm.  Then one can further reduce to just recovering~$p(0^n)$, using the altered-channel trick.
As for learning $p(0^n)$, we first observe that the replacement of $C$ by $C^{\neq 0^n}$ changes nothing for this problem, so we effectively have the same task just for the $\frac13$-crossover $Z$-channel on binary strings.
This is similar to the erasure channel with erasure probability $\frac13$, and in fact the known solutions for erasure probability-$r$~\cite{DRWY12,MS13,DOS20,PSW17} \emph{only use the locations of the $\one$'s in the received word}.
Thus these known solutions work \emph{equally well} for the $Z$-channel.
Indeed, as noted in~\cite{DRWY12}, the solution is particularly simple when $r \leq \frac12$ (as it is for us); the full method of ``robust local inverses'' is not needed, and one can use the ``natural inverse'' (as we implicitly do in the proof of \Cref{thm:individ} below).

\subsection{Individual Recovery}    \label{sec:individ}

Although the proof of the below theorem is self-contained, we remark that it implicitly follows the Individual Recovery routine of~\cite{DRWY12} for the $\frac13$-erasure channel.
\begin{theorem}                                     \label{thm:individ}
    For any fixed $B \in \{0,1,2,3\}^n$, a version of \Cref{thm:intro-pop} holds in which the learner only computes an estimate $\widehat{p}(B)$ of $p(B)$ satisfying $|\widehat{p}(B) - p(B)| \leq \eps_0$ except with probability at most~$\delta_0$.
    The number of samples used is $m = O(1/\eps_0^2) \cdot \log(1/\delta_0)$ and the classical post-processing time is $O(mn)$.
\end{theorem}
\begin{remark}
    The reader may wish to verify the proof just in the case~$B = 0^n$, where it is simpler; the general case then     follows from \Cref{obs:alter-trick}.
\end{remark}
\begin{proof}
    Alice obtains $m$ probe/readout pairs $(\bA, \bR)$, with $\bA \sim \{1,2,3\}^n$ uniformly random and $\bR = \bA \star \bC$, where $\bC$ is a random channel outcome drawn from $p$. 
    The estimate $\widehat{p}(B)$ that Alice will output is the empirical mean of the random variable
    \[
        \bH = (-1/2)^{|\bA \star B +_2 \bR|} = (-1/2)^{\sum_{t} ((\bA \star B) +_2 \bR)_t} = \prod_{t=1}^n (-1/2)^{\by_t}, \quad \by_t \coloneqq  (\bA_t \star B_t) +_2 \bR_t.
    \]
    As seen in \Cref{obs:altered-z}, for a given outcome $\bC = C$, the random binary string $(\bA \star B) +_2 \bR$ is distributed as $C^{\neq B}$ passed through a $Z$-channel with crossover probability~$\frac13$.
    In particular, its coordinates $\by_t$ are independent random variables, with conditional expectation given by
    \[
        \E[(-1/2)^{\by_t} \mid \bC = C] = \begin{cases}
                                                                            (-1/2)^0 = 1 & \text{if $C_t = B_t$,} \\
                                                                            \tfrac13 (-1/2)^0 + \tfrac23 (-1/2)^1 = 0 & \text{if $C_t \neq B_t$}.
                                                                        \end{cases}
    \]
    Thus
    \[
        \E[\bH \mid \bC = C] = \prod_{t=1}^n \E\bigl[(-1/2)^{\by_t} \mid \bC = C\bigr] = \begin{cases}
            1 & \text{if $C = B$,} \\
            0 & \text{if $C \neq B$,}
        \end{cases}
    \]
    and hence indeed $\E[\bH] = p(B)$.
\end{proof}

\subsection{Population Recovery}            \label{sec:population}
\Cref{thm:individ} allows Alice to estimate $p(B)$ for any particular string~$B \in \{0,1,2,3\}^n$.
But also, for any shorter string $\beta \in \{0,1,2,3\}^\ell$, Alice can estimate the marginal
\[
    p(\beta) \coloneqq \sum_{\gamma \in \{0,1,2,3\}^{n-\ell}} p(\beta\gamma) = \Pr_{C \sim p}[(C_1, \dots, C_\ell) = \beta],
\]
simply by ignoring all data in positions $\ell+1, \dots, n$. 
(She is obviously not limited to marginalizing contiguous blocks, but this is all we will need for our purposes.)
Alice can thus learn all of~$p$ to good $\ell_\infty$-precision with the straightforward, coordinate-by-coordinate branch-and-prune approach common in Population Recovery (see, e.g.,~\cite[App.~A]{PSW17}).
We repeat this approach here; the following algorithm achieves our main \Cref{thm:intro-pop} for Pauli channels, except for the claim about the running time of the post-processing algorithm:

\begin{enumerate}
    \item Set $\eps_0 = \frac{\eps}{4}$, $\delta_0 = \frac{4\eps \delta}{9n}$ and draw a single batch of~$m$ samples, where $m$ is as in \Cref{thm:individ}.

    \item Define ``support sets'' $\Omega_1 = \{0, 1, 2, 3\}$ and $\Omega_2 = \cdots = \Omega_n = \emptyset$.

    \item For round $j = 1 \dots n-1$:

    \item \quad For each prefix $\beta' \in \Omega_{j}$ and each $b \in \{0, 1, 2, 3\}$:

    \item \quad \quad Run the Individual Recovery algorithm on $\beta \coloneqq \beta'b$ to estimate the marginal $p(\beta)$.

    \item \quad \quad If the estimate is at least $2\eps_0 = \frac{\eps}{2}$, then place $\beta$ into $\Omega_{j+1}$.

    \item Output as $\widehat{p}$ the collection of strings in $\Omega_n$, together with their estimated probabilities.
\end{enumerate}

The correctness of the algorithm, that $\|\widehat{p} - p\|_\infty \leq \eps$  with failure probability at most~$\delta$, is straightforward and is explicitly proven in~\cite[Lem.~18]{PSW17}.
The proof also establishes that when there is no failure, $|\Omega_j| \leq \frac{4}{\eps}$ holds for all $1 \leq j \leq n$.
Thus for running time purposes (and without impacting the correctness claim) we may have the algorithm abort if ever some~$\Omega_j$ gets cardinality more than~$\frac{4}{\eps}$.
It only remains to obtain the post-processing running time of $O(mn/\eps)$ claimed in \Cref{thm:intro-pop}.

\paragraph{Running time analysis.}
As it stands, the running time of the above algorithm is $O(mn^2/\eps)$, since it may do up to $O(n/\eps)$ executions of the $O(mn)$-time Individual Recovery algorithm. 
(We have implemented this na\"{\i}ve version of the algorithm in Julia~\cite{F21} for interested readers.) 
However, since all executions of the Individual Recovery algorithm are on the same batch of samples, it's not hard to see that information from the $j$th round of the algorithm can be used to speed up the $(j+1)$st round.
More precisely, we show that each round can be done in $O(m/\eps)$ time, leading to the overall claimed running time of $O(mn/\eps)$.

Let $R \in \{\zero, \one\}^{m \times n}$ be the measurement outcome bits that the algorithm processes, and let $R_{1 \dots j}$ denote the submatrix formed by the first~$j$ columns.
Also, for $\beta \in \{0,1,2,3\}^j$, let $R^{(\beta)} \in \{\zero, \one\}^{m \times j}$ be the (hypothetical) matrix whose $t$th row is the same as $R_{1 \dots j}$'s but with $(A^t_1, \dots, A^t_j) \star \beta$ added in mod~$2$, where $A^t$ is the $t$th probe string used by Alice.
Given $\beta$, the algorithm can look up entries of $R^{(\beta)}$ in~$O(1)$ time.

Recall that when the algorithm does Individual Recovery on the prefix~$\beta$, it computes the fraction of rows of $R^{(\beta)}$ that have Hamming weight~$i$, multiplies this number by~$(-1/2)^i$, and sums the results.
In particular, this estimate can be computed in~$O(m)$ time given the vector $h^{(\beta)} \in \N^m$ whose $t$th entry is the Hamming weight of the $t$th row of $R^{(\beta)}$ --- just add up $(-1/2)^{h^{(\beta)}_t}/m$ across all~$t$.

We can now modify the above Population Recovery algorithm so that whenever a prefix $\beta \in \{0,1,2,3\}^j$ is added into~$\Omega_j$, the algorithm retains the vector~$h^{(\beta)}$ that went into estimating~$p(\beta)$.
It is easy to see that in the subsequent round, we can compute each of $h^{(\beta0)}, h^{(\beta1)}, h^{(\beta2)}, h^{(\beta3)}$ from $h^{(\beta)}$ (and hence the marginal estimates) in~$O(m)$ time, and retain them as needed.
Thus indeed each round only requires $O(m/\eps)$ time, since at most $\frac{4}{\eps}$ prefixes are processed in each round.

\section{Multiplicative error}  \label{sec:mult}
In a practical scenario we would would hope that the ``nontrivial error rate'' of the Pauli channel,
\[
    \eta \coloneqq 1 - p(0^n)
\]
is very small.
This motivates writing $p$ as a mixture distribution, as follows:
\begin{equation}    \label{eqn:mix}
    p:  \quad \text{mixing weight $1-\eta$ on } 0^n, \quad \text{mixing weight $\eta$ on } \perr,
\end{equation}
where $\perr$ is a distribution on $\{0,1,2,3\}^n \setminus \{0^n\}$.
Now a natural goal is to learn with \emph{multiplicative error}~$\eps$, meaning producing estimates~$\widehat{\eta}$,~$\estperr$ with
\[
    (1-\eps)\eta \leq \widehat{\eta} \leq (1+\eps)\eta, \qquad \|\estperr - \perr\|_\infty \leq \eps.
\]
As described in \Cref{sec:intro}, the ideal sample complexity to strive for now is~$O(\frac{1}{\eps^2 \eta})$.

\paragraph{Adaptivity, and a floor on~$\eta$.}
Let us make two more technical remarks.
First, if $\eta$ is extraordinarily small (or even~$0$), we won't want to make $1/\eta$ measurements.
Thus we assume the algorithm is given a floor~$\eta_0$, and when $\eta \leq \eta_0$ we are satisfied just to certify that this is the case.
Second, we cannot hope to have (as before) a completely nonadaptive algorithm achieving sample complexity on the order of~$1/(\eps^2 \eta)$ because the algorithm does not know~$\eta$, or even an approximation to~$\eta$, in advance. 
Thus our algorithm will first need to find a preliminary constant-factor approximation $\etaest$ to~$\eta$ in an online probe-and-measure fashion; then it can proceed nonadaptively.

\subsection{Roughly estimating the error rate}
Here we describe the (mildly) ``adaptive'' algorithm that handles the error floor and obtains~$\etaest$, a factor-$5$ approximation of~$\eta$ before subsequently finding a good approximation to all the error rates (including $p(0^n) = 1-\eta$).

\begin{lemma} \label{lem:approx-eta}
    There is a randomized learning algorithm that, given input $0 < \delta_0, \eta_0 < 1$, as well as access to an $n$-qubit Pauli channel defined by distribution~$p$ with nontrivial error rate $\eta = 1 - p(0^n)$:
    \begin{itemize}
        \item repeatedly prepares a state, passes it through the Pauli channel, and measures, as in \Cref{thm:intro-pop};
        \item halts after some number of repetitions (\emph{always} at most $O(1/\eta_0)\cdot \log(1/\delta_0)$) and outputs either: ``$\eta \leq \eta_0$'' or else an estimate $\etaest$ that is within a factor of~$5$ of~$\eta$;
        \item runs in classical time that is linear in the number of measurement readouts.
    \end{itemize}
    Except with probability at most~$\delta_0$, the algorithm's output is correct and it halts after at most $O(1/\eta) \cdot \log(1/\delta_0)$ repetitions.
\end{lemma}
\begin{proof}
    Recall \Cref{fact:z-channel2}: by doing random nontrivial probes, an algorithm can get samples from a random string that is non-$\zero^n$ with some probability $\eta'$ between $\frac23 \eta$ and~$\eta$.
    In order to find the factor-$5$ approximation $\etaest$ of $\eta$, it suffices for the algorithm to estimate $\eta'$ up to a factor of~$3$ or else certify $\eta' \leq \eta_0$.
    This is now a completely standard problem: estimating the bias of an $\eta'$-biased coin up to a factor of~$3$ using on the order of $1/\eta'$ flips, despite not knowing $\eta'$ in advance.
    The algorithm is the obvious one: repeatedly flip until getting ``heads'' (but never more than $O(1/\eta_0)$ times), convert the number of flips~$\bG$ into the estimate~$1/\bG$, then take the median of $O(\log(1/\delta))$ estimates.
    We omit the straightforward classical analysis.
\end{proof}

\subsection{Individual Recovery with multiplicative error}
We henceforth assume the algorithm from \Cref{lem:approx-eta} succeeded and that $\etaest$ is a factor-$5$ approximation of the true error rate~$\eta$.
We now describe how the algorithm can do ``Individual Recovery'' with multiplicative error.
A note: the sample complexities are stated in terms of the parameter~$\eta$; formally, the algorithm does not know~$\eta$, but it can use $5\etaest$ (which it knows) in its place, and the $O(\cdot)$ bounds are not affected.

We first show that the algorithm from \Cref{thm:individ} already achieves the desired multiplicative-error/sample tradeoff in the case of estimating~$\eta$:

\begin{proposition} \label{prop:est-error}
    Given $\etaest$ within a factor $5$ of~$\eta = 1-p(0^n)$, a version of \Cref{thm:individ} holds in which, for $B = 0^n$, the estimate $\widehat{p}(0^n)$ satisfies $|\widehat{p}(0^n) - p(0^n)| \leq \eps \eta$ except with failure probability at most~$\delta_0$, and the number of samples used is $m = O(\frac{1}{\eps^2 \eta}) \cdot \log(1/\delta_0)$.
\end{proposition}
\begin{remark}
    The success event here is equivalent to the estimate $\widehat{\eta} = 1-\widehat{p}(0^n)$ satisfying the inequality $(1-\eps) \eta \leq \widehat{\eta}\leq (1+\eps)\eta$.
\end{remark}
\begin{proof}
    The algorithm used is the same as the one in \Cref{thm:individ} (with $B = 0^n$); only the analysis changes.
    Recall that the algorithm's estimate is the empirical mean of $\bH = (-1/2)^{|\bR|}$, a random variable whose true mean is~$p(0^n) = 1-\eta$.
    Equivalently we may consider the random variable $\overline{\bH} = 1 - \bH$, which has true mean~$\eta$ and which is supported in~$[0,2]$.
    But now a standard multiplicative Chernoff bound shows that the empirical mean $\widehat{\eta}$ of~$\overline{\bH}$ after $O(1/(\eps^2\eta)) \cdot \log(1/\delta_0)$ samples indeed satisfies $(1-\eps) \eta \leq \widehat{\eta}\leq (1+\eps)\eta$.
\end{proof}

\begin{proposition}
    A trivial modification of \Cref{thm:individ} also achieves, for any $B \neq 0^n$, an estimate $\widehat{p}(B)$ satisfying $|\widehat{p}(B) - p(B)| \leq \eps \eta$ except with failure probability at most~$\delta_0$, using $m = O(\frac{1}{\eps^2 \eta}) \cdot \log(1/\delta_0)$ samples.
\end{proposition}
\begin{proof}
    Rather than empirically estimating the mean of $\bH = (-1/2)^{|\bA \star B +_2 \bR|}$, the algorithm instead empirically estimates the mean of $\bH' = \bH - (-1/2)^{|\bA \star B|}$, a random variable bounded in~$[-2,2]$.
    (Note that Alice knows $B$ and also each probe string $\bA$, hence can compute $(-1/2)^{|\bA \star B|}$ herself.)
    It is easy to see that $\E[ (-1/2)^{|\bA \star B|}] = 0$ using $B \neq 0^n$.
    Thus $\bH'$ remains an unbiased estimator for $p(B)$; i.e., $\E[\bH'] = p(B)$.
    But furthermore note that $\bH'$ is almost always~$0$; specifically, whenever the channel outcome $\bC$ is~$0^n$ (probability $1-\eta$), we have $\bR = \bA \star 0^n = \zero^n$ and hence $\bH' = (-1/2)^{|\bA \star B|} - (-1/2)^{|\bA \star B|} = 0$.
    Thus using $|\bH'| \leq 2$ we trivially conclude $\E[(\bH')^2] \leq 4\eta$.
    But now it follows from the Bernstein inequality (see, e.g.,~\cite[Ch.~2, Prop.~2.4]{Wai19}) that to estimate the mean of a random variable~$\bH'$ that is bounded in~$[-2,2]$ and has~$\E[(\bH')^2] = s$, it suffices to use $\frac{s+2\gamma/3}{\gamma^2} \ln(2/\delta_0)$ samples to achieve additive error~$\gamma$ except with probability at most~$\delta_0$.
    Thus taking $\gamma = \eps \eta$ and using $s \leq 4\eta$ indeed completes the proof.
\end{proof}

\subsection{Population Recovery with multiplicative error}  

Combining the results from the previous section on Individual Recovery with the reduction in \Cref{sec:population} immediately proves our \Cref{thm:intro-pop2} (in the case of Pauli channels). 

\section{Further extensions: general channels and measurement noise}    \label{sec:extend}

\subsection{Pauli error rates of general channels}  \label{sec:general-chan}
With very minor effort we can now upgrade our algorithm to learn the ``Pauli error rates'' of a \emph{general} quantum channel, thereby fully establishing our \Cref{thm:intro-pop}.

We recall the following definitions/facts (see, e.g.,~\cite[Lem.~5.2.4]{Dan15}):
\begin{definition}
    Let $\Lambda$ denote an arbitrary $n$-qubit quantum channel.
    Its \emph{Pauli twirl} $\Lambda_P$ is the $n$-qubit quantum channel defined by
    \[
        \Lambda_P \rho = \E_{\bT \sim \{0,1,2,3\}^n} [\sigma_{\bT}^\dagger (\Lambda \sigma_{\bT} \rho \sigma^\dagger_{\bT}) \sigma_{\bT}].
    \]
    The channel $\Lambda_P$ is itself a Pauli channel; the associated probabilities $p(C)$ are called the \emph{Pauli error rates} of~$\Lambda$.
\end{definition}
\begin{fact}
    Suppose we write $K_j$ for the Kraus operators of~$\Lambda$, so $\Lambda \rho = \sum_j K_j \rho K_j^\dagger$.
    Further suppose that $K_j$ is represented in the Pauli basis as $K_j = \sum_{C \in \{0,1,2,3\}^n} \alpha_{j,C} \sigma_C$.
    Then $\Lambda$'s Pauli error rates are given by $p(C) = \sum_j |\alpha_{j,C}|^2$.
\end{fact}

It's easy to see that, given access to a general channel~$\Lambda$, a learner Alice can simulate access to its Pauli twirl~$\Lambda_P$: whenever Alice wishes to pass $\rho$ through $\Lambda_P$, she instead chooses $\bT \sim \{0,1,2,3\}^n$ uniformly at random, passes $\sigma_{\bT} \rho \sigma_{\bT}^\dagger$ through~$\Lambda$,  and replaces the channel output~$\tau$ with $\sigma_{\bT}^\dagger \tau \sigma_{\bT}$.

In our context of learning Pauli error rates, this simulation becomes particularly simple.
Recall that our algorithm for Pauli channels only ever passes pure states of the form $\ket{\chi^{A_1}_{+}} \ket{\chi^{A_2}_{+}} \cdots \ket{\chi^{A_n}_{+}}$ through the  channel, for $A \in \{1,2,3\}^n$.
Further, the channel output is always measured in the associated Pauli bases, the $j$th qubit of the output measured in the basis $\ket{\chi^{A_n}_{\pm}}$.
The effect of simulating the Pauli twirl with $\sigma_{\bT}$ is simply to replace the input $\ket{\chi^{A_j}_{+}}$ to qubit $j$ with the input $\ket{\chi^{A_j}_{(-1)^{A_j \star \bT_j}}}$, and to add $A \star \bT$ to the measurement outcomes.
Thus we may deduce the full version of \Cref{thm:intro-pop2} (concerning learning Pauli error rates of general channels) from the already-established special case of learning Pauli channels.

\subsection{Tolerating measurement errors} \label{sec:tolerating}
It is also straightforward to see that our algorithm can tolerate a mild form of measurement error.
Suppose that we have an imperfect $1$-qubit measuring device that is used to implement the three Pauli-basis measurements. 
More precisely, we assume it has the following property:
When applied to a qubit in a Pauli eigenvalue state, the measuring device ``fails'' (say, reads out ``$\texttt{?}$'') with probability~$\nu$, and otherwise behaves ideally.
Here~$\nu$ is a parameter that we assume is known to the learner through estimation, and that measurement failures are independent events.

As discussed in the paragraph just preceding \Cref{sec:individ}, our algorithm for estimating any $p(B)$ is effectively performing the standard ``Individual Recovery algorithm'' for the \emph{binary erasure channel} with erasure probability~$\frac13$.
(Recall that we actually have the $Z$-channel with crossover probability~$\frac13$ applied to the binary string $C^{\neq B}$, but that the erasure channel algorithm only uses the locations of the $\one$'s in the received string, and thus works equally well for the $Z$-channel.)
The effect of measuring device failures is to replace the erasure probability~$\frac13$ with $r \coloneqq \nu + (1-\nu)\frac13$.
So long as $r \leq \frac12$, the standard recovery algorithm for probability $r$-erasures works just as well~\cite{DRWY12}: the only change needed is that the factor ``$(-1/2)$'' appearing in \Cref{thm:individ}'s definition of~$\bH$ needs to be replaced by $-r/(1-r)$.
(Note that this quantity has magnitude bounded by~$1$ if and only if~$r \leq \frac12$.)
But the condition $r \leq \frac12$ is equivalent to $\nu \leq \frac14$, and this justifies our \Cref{thm:pop4}.

(In fact, for erasure probability $\frac12 < r < 1$, much more sophisticated algorithms~\cite{DOS20,PSW17} can succeed at Individual Recovery, at the expense of increasing the sample complexity from the order of $1/\eps^2$ to the order of~$1/\eps^{2r/(1-r)}$; but for simplicity, we ignore pursuing this extension.)

\section{An alternative, Fourier approach} \label{sec:fourier}
Here we give an alternative algorithm for learning Pauli channels, using a perspective from Boolean Fourier analysis; see~\cite[Chaps.~1, 3]{OD14} for background and notation.

For Pauli channels, the $\F_2$-Fourier transform relates the error rates and the channel eigenvalues. 
The Pauli operators themselves are the eigenvectors of a Pauli channel, and we can easily compute the eigenvalue associated to~$\sigma_A$  using the relation $\sigma_A \sigma_C = (-1)^{A\star C}\sigma_C \sigma_A$ via
\[
    \sigma_A \mapsto \sum_{C \in \{0,1,2,3\}^n} p(C) \cdot \sigma_C \sigma_A \sigma_C^\dagger = \sum_{C \in \{0,1,2,3\}^n} p(C) \cdot (-1)^{\sum_{i=1}^n (A \star C)_i} \sigma_A = \lambda_A \sigma_A,
\]
so that $\lambda_A = \E_{\bC \sim \{0,1,2,3\}^n}\bigl[ 2^{2n} p(\bC) \cdot (-1)^{\sum_{i=1}^n (A \star \bC)_i}\bigr]$. 
This clearly resembles an $\F_2$-Fourier transform. 

To make this connection more explicit, in the remainder of this section we will identify the elements of $\{0,1,2,3\}$ with their base-$2$ representations in~$\F_2^2$.
Let us use overline to denote the swapping operation on two bits; i.e., $\overline{a_1a_2} = a_2a_1$ for $a_1,a_2 \in \F_2$.
We extend the notation $n$-fold to vectors $A \in \F_2^{2n} \cong (\F_2^2)^n$.
(Equivalently, we have $\ol{0} = 0$, $\ol{1} = 2$, $\ol{2} = 1$, $\ol{3} = 3$, and we extend the notation coordinate-wise to $A \in \{0,1,2,3\}^n$.)
Now define the \emph{symplectic dot product}
\[
    \la A, C \ra = \ol{A} \cdot C = \sum_{i=1}^n (A \star C)_i \mod 2,
\]
where $A \cdot C$ denotes the usual dot product on~$\F_2^{2n}$. 
A Pauli channel eigenvalue is now equivalently written in two ways as 
\begin{equation*}  
    \lambda_A = \E_{\bC \sim \F_2^{2n}}\bigl[ 2^{2n} p(\bC) \cdot (-1)^{\langle A,\bC\rangle}\bigr] = \E_{\bC \sim \F_2^{2n}}\bigl[ 2^{2n} p(\bC) \cdot (-1)^{\ol{A}\cdot\bC}\bigr].
\end{equation*}
Let us write $\varphi$ for the probability \emph{density} (vis-a-vis the uniform distribution) associated to~$p$; i.e., $\varphi(C) = 2^{2n} p(C)$. 
Then the Fourier transform $f = \tilde{\varphi}$ is given by
\begin{equation}    \label{eqn:fourier1}
    f(A) = \tilde{\varphi}(A) = \E_{\bC \sim \F_2^{2n}}[\varphi(\bC) (-1)^{A \cdot \bC}] = \E_{\bC \sim p}[(-1)^{A \cdot \bC}] = \lambda_{\ol{A}}.
\end{equation}
Observe that $f$ (and equivalently $\lambda$) are functions $f : \F_2^{2n} \to [-1,1]$ and that $p = \tilde{f}$. 
Such group character averages were considered in the context of quantum noise estimation in~\cite{Helsen2019}. 
While we can talk interchangeably about the Fourier coefficients of the density $\varphi$ and the channel eigenvalues $\lambda$ (as they are related by $f(A) = \lambda_{\ol{A}}$), we will focus on $f$ in what follows. 

We see from \Cref{eqn:fourier1} that
\begin{equation}    \label{eqn:estim}
    f(A) = \E_{\bC \sim p}[(-1)^{\la \ol{A}, C \ra}] = \E_{\bC \sim p}[(-1)^{\sum_t (\ol{A}  \star C)_t}],
\end{equation}
and as we now describe this means Alice can straightforwardly estimate $f(A)$ for any~$A$ of her choosing.

Let's extend \Cref{def:nontrivial-probe} of ``nontrivial probe'' to allow not just for $A \in \{1,2,3\}^n$ but any $A \in \{0,1,2,3\}^n$; we omit the adjective ``nontrivial'' in this more general case.
To handle coordinates~$j$ where $A_j = 0$, Alice can simply put any qubit $\ket{\chi}$ into the $j$th position of her state $\ket{\psi_A}$, ignore the $j$th position coming out of the channel, and automatically treat the $j$th readout bit as~$\zero$.
In this way, \Cref{fact:readout} still holds: for any probe $A \in \{0,1,2,3\}^n$ and any string $C \in \{0,1,2,3\}^n$ drawn by Charlie, the readout is $R = A \star C \in \{\zero,\one\}^n$.
It follows that Alice can empirically estimate the right-hand side of \Cref{eqn:estim} by repeatedly probing the channel with~$\ol{A}$ and averaging the following function of $\bR$, the readout:~$(-1)^{\sum_t \bR_t}$.
This yields $f(A)$ to additive precision~$\eps$ with confidence at least $1-\delta$, using  $O(1/\eps^2) \cdot \log(1/\delta)$ probes; we refer to this as ``efficient estimation''.

We now see that Alice has (noisy) query access to~$f : \F_2^{2n} \to [-1,1]$, and her goal is to estimate the large values of $p = \tilde{f}$.
This task is highly reminiscent of the task solved by the Goldreich--Levin learning algorithm~\cite{GL89}.
The minor differences are that Goldreich--Levin typically assumes \emph{perfect} query access to some~$f : \F_2^{2} \to \{-1,1\}$, and has the normalization that $\sum_C \tilde{f}(C)^2 = 1$, rather than our normalization of $\sum_C \tilde{f}(C) = \sum_C p(C) = 1$.
Still, if one ``unrolls'' the Goldreich--Levin algorithm in this context, one gets almost the same solution for learning Pauli channels as described in \Cref{sec:population}: reduction from Population Recovery to Individual Recovery.

\subsection{The Goldreich--Levin approach}
In a typical exposition of the Goldreich--Levin algorithm (e.g.~\cite[Ch.~3.5]{OD14}, which we'll follow), one assumes Alice has perfect query access to an $f : \F_2^n \to \{-1,1\}$.
Herein we sketch the alterations to this exposition that are needed for learning Pauli channels. 
We note that a ``quantum Goldreich--Levin'' algorithm was given by Montanaro and Osborne~\cite{MontanaroOsborne2010} for learning the class of quantum boolean functions, which includes the unitary Pauli channels, but this makes explicit use of the unitary property and hence doesn't immediately apply to general Pauli channels. 

One basic subroutine in the Goldreich--Levin algorithm (akin to ``Individual Recovery'')  is using query access to~$f$ to efficiently estimate $\tilde{f}(B)$ for various~$B$.
This is done (see~\cite[Prop.~3.30]{OD14}) via straightforward empirical estimation:
\begin{equation}    \label{eqn:basic-Fourier}
    \tilde{f}(B) = \E_{\bA \sim \F_2^{2n}}[f(\bA) (-1)^{\bA \cdot B}].
\end{equation}
Recall that in our setting, Alice can only access $f(\bA)$ by empirically estimating it via \Cref{eqn:estim}.
Inserting this into the above, we get
\[
    \tilde{f}(B) = \E_{\bA \sim \F_2^{2n}} \E_{\bC \sim p}[(-1)^{\sum_t (\ol{\bA}  \star \bC)_t + \bA \cdot B}].
\]
Thus as needed in Goldreich--Levin, Alice can efficiently estimate this for any~$B$ of her choosing by picking uniformly random~$\bA \in \{0,1,2,3\}^n$, probing the channel with~$\ol{\bA}$, and averaging the following function of~$\bR$, the readout: $(-1)^{\sum_t \bR_t + \bA \cdot B}$.
Indeed the reader will note that this method is almost the same as the one used in \Cref{thm:individ}!
The essential difference is that $\bA$ is uniform on $\{0,1,2,3\}^n$ rather than $\{1,2,3\}^n$, which effectively makes the ``crossover probability''~$\frac12$ instead of~$\frac13$, and hence the factor $(-1/2) = -\frac{1/3}{1-1/3}$ becomes $(-1) = -\frac{1/2}{1-1/2}$.
Note that this difference implies that the Goldreich--Levin approach does not immediately tolerate measurement failures as in \Cref{sec:tolerating}. 

As mentioned earlier, Goldreich--Levin typically assumes $f : \F_2^n \to \{-1,1\}$ and hence we have $\sum_{C \in \F_2^n} \tilde{f}(C)^2 = 1$; its goal is to find all $B$ with $|\tilde{f}(B)| \geq \eps$, knowing that there are automatically at most $1/\eps^2$ such~$B$.
It accomplishes this via a ``branch-and-prune'' strategy that relies on the ability to estimate $\sum_{C' \in \F_2^{n-k}} \tilde{f}(\beta, C')^2$ for any prefix $\beta \in \F_2^k$.
In our setup, with $f : \F_2^{2n} \to [-1,1]$, we instead know a priori that $p = \tilde{f}$ satisfies $\sum_C \tilde{f}(C) = 1$, and our goal is to find all $B$ with $|\tilde{f}(B)| \geq \eps$.
Thus the search is even easier than in Goldreich--Levin, as the same branch-and-prune strategy works with non-squared Fourier coefficients.
Following the strategy gives the same Population-to-Individual Recovery algorithm as in \Cref{sec:population}.

\subsection{Final remarks}
As mentioned earlier, the techniques used in the previous works~\cite{FW20,Harper2020,Harper2021} on Pauli channel estimation are Fourier-based.
The paper~\cite{FW20} achieves SPAM tolerance, and manages to trade some measurement complexity for channel-reuse; on the other hand, its bounds have a dependency on the channel eigenvalue gap~$\Delta = \min_{A \neq 0^n} \{1-|\lambda_A|\}$, which may be arbitrarily small. 
As shown in the previous section, one can recover our (SPAM-less) Pauli estimation results via the Fourier approach with no dependence on~$\Delta$ and without assumptions about the noise or the support.

We believe that it is possible to obtain a common generalization of the results in~\cite{FW20} and the present paper that achieves the best of both worlds via this Fourier approach: SPAM-robust and efficient Pauli channel estimation with no dependence on $\Delta$.
We leave this for future work.

\acknowledgements

We thank Robin Harper for discussions about Pauli channels.
This work was supported by ARO grant W911NF2110001. 
R.O.\ is additionally supported by NSF grant FET-1909310. 
This material is based upon work supported by the National Science Foundation under grant numbers listed above.
Any opinions, findings and conclusions or recommendations expressed in this material are those of the author and do not necessarily reflect the views of the National Science Foundation (NSF).

%\bibliographystyle{habbrvdoi}
%\bibliography{Pauli_v2}

\end{document}